\documentclass[11pt]{amsart}

\usepackage[centertags]{amsmath}
\usepackage{amsthm,amsfonts}

\usepackage{tikz}
\usetikzlibrary{positioning}

\usepackage[toc,page]{appendix}

\usepackage{amssymb}
\usepackage{epsfig}
\usepackage{amssymb, latexsym, mathrsfs}
\usepackage{indentfirst}

\usepackage{color}
\usepackage{ulem}
\usepackage{verbatim}
\begin{document}
\theoremstyle{plain}
\newtheorem{theorem}{Theorem}[section]
\newtheorem{lemma}[theorem]{Lemma}
\newtheorem{result}[theorem]{Result}
\newtheorem{hw}[theorem]{Homework}
\newtheorem{corollary}[theorem]{Corollary}
\newtheorem{proposition}[theorem]{Proposition}
\newtheorem{example}[theorem]{Example}
\newtheorem{examples}[theorem]{Examples}
\newtheorem{conjecture}[theorem]{Conjecture}
\theoremstyle{definition}
\newtheorem{notations}[theorem]{Notations}
\newtheorem{notation}[theorem]{Notation}
\newtheorem{remark}[theorem]{Remark}
\newtheorem{remarks}[theorem]{Remarks}
\newtheorem{definition}[theorem]{Definition}
\newtheorem{claim}[theorem]{Claim}
\newtheorem{assumption}[theorem]{Assumption}
\newtheorem{question}[theorem]{Question}
\numberwithin{equation}{section}
\newtheorem{construction}[theorem]{Construction}
\newtheorem{examplerm}[theorem]{Example}
\newtheorem{propositionrm}[theorem]{Proposition}
\newcommand{\binomial}[2]{\left(\begin{array}{c}#1\\#2\end{array}\right)}
\newcommand{\zar}{{\rm zar}}
\newcommand{\an}{{\rm an}}
\newcommand{\red}{{\rm red}}
\newcommand{\codim}{{\rm codim}}
\newcommand{\rank}{{\rm rank}}
\newcommand{\Pic}{{\rm Pic}}
\newcommand{\Div}{{\rm Div}}
\newcommand{\Hom}{{\rm Hom}}
\newcommand{\Ima}{{\rm Im}}
\newcommand{\Ker}{{\rm Ker \ }}
\newcommand{\Spec}{{\rm Spec}}
\newcommand{\sing}{{\rm sing}}
\newcommand{\reg}{{\rm reg}}
\newcommand{\Char}{{\rm char\ }}
\newcommand{\Tr}{{\rm Tr}}
\newcommand{\Norm}{{\rm Norm}}
\newcommand{\tr}{{\rm tr}}
\newcommand{\supp}{{\rm supp \ }}
\newcommand{\Gal}{{\rm Gal}}
\newcommand{\Span}{{\rm Span}}
\newcommand{\Min}{{\rm Min \ }}
\newcommand{\Max}{{\rm Max \ }}

\newenvironment{pf}{\noindent\textbf{Proof.}\quad}{\hfill{$\Box$}}
\newenvironment{dpf}{\noindent\textbf{Disproof.}\quad}{\hfill{$\Box$}}

\newcommand{\sC}{{\mathcal C}}
\newcommand{\sD}{{\mathcal D}}
\newcommand{\sK}{{\mathcal K}}
\newcommand{\sL}{{\mathcal L}}
\newcommand{\sR}{{\mathcal R}}
\newcommand{\sS}{{\mathcal S}}

\newcommand{\C}{{\mathbb C}}
\newcommand{\F}{{\mathbb F}}
\newcommand{\Q}{{\mathbb Q}}
\newcommand{\R}{{\mathbb R}}
\newcommand{\Z}{{\mathbb Z}}

\newcommand{\be}{\begin{eqnarray}}
\newcommand{\ee}{\end{eqnarray}}
\newcommand{\nn}{{\nonumber}}
\newcommand{\dd}{\displaystyle}
\newcommand{\ra}{\rightarrow}

\newcommand{\Keywords}[1]{\par\noindent
{\small{\bf Keywords\/}: #1}}

\title[Additive cyclic codes over finite commutative chain rings]{Additive cyclic codes \\ over finite commutative chain rings}
\author[Edgar Mart\'inez-Moro, Kamil Otal and Ferruh \"{O}zbudak]{Edgar Martinez-Moro, Kam\.{I}l Otal and Ferruh \"{O}zbudak}
\thanks{Edgar Mart\'inez-Moro is with the Mathematics Research Institute, University of Valladolid, Castilla, Spain. Partially supported by MINECO MTM2015-65764-C3-1-P project and MTM2015-69138-REDT Network}
\thanks{Kamil Otal and Ferruh \"{O}zbudak are with Department of Mathematics and Institute of Applied Mathematics, Middle East Technical
University,  Dumlup{\i}nar Bulvar{\i} No. 1, 06800, Ankara, Turkey;
e-mail: \{kotal,ozbudak\}@metu.edu.tr}

\abstract

Additive cyclic codes over Galois rings were investigated in \cite{CGFC2015}. In this paper, we investigate the same problem but over a more general ring family, finite commutative chain rings. When we focus on non-Galois finite commutative chain rings, we observe two different kinds of additivity. One of them is a natural generalization of the study in \cite{CGFC2015}, whereas the other one has some unusual properties especially while constructing dual codes. We interpret the reasons of such properties and illustrate our results giving concrete examples.  
\vspace{0.7cm}
\endabstract
\maketitle
\section{Introduction}\label{sec:intro}

Additive codes are a direct and useful generalization of linear codes, and they have applications in quantum error correcting codes. There are several studies using different approaches on them and their applications (see, for example, \cite{Bie2007, CRSS1998, Huf2013, Rai1999}).  

Cyclic codes are one of the most attractive code families thank to their rich algebraic structure and easy implementation properties. There are many generalizations of cyclic codes in different directions, e.g., \cite{DL2004, LOOS2013, MPR2013, MR2006}.  

Codes over rings have been of interest in the last quarter after the discovery that some linear codes over $\Z_4$ are related to non-linear codes over finite fields (see, for example, \cite{CHKSS1993, CS1995, HKCSS1994, Nec1982, Nec1991}). The first family of the rings used in this perspective was $\Z_{p^n}$, where $p$ is a prime and $n$ is a positive integer. The most important property of such rings is the linearity of their ideals under inclusion. Therefore, the generalizations on Galois rings, or moreover finite chain rings are immediate. Some recent works on codes over such rings are \cite{CGFC2015, DL2004, HL2010, SM2016}. 

Finite chain rings, besides their practical importance, are quite rich mathematical objects and so they have also theoretical attraction. They have connections in both geometry (Pappian Hjelmslev planes) and algebraic number theory (quotient rings of algebraic integers). These connections have also been interpreted in applications (as an example of application in coding theory, see \cite{HL2010}). Some main sources about finite commutative chain rings in the literature are \cite{BF2002, McD1974, Nec2008, Wan2003}. 

\subsection{Related work and our contribution}

Additive cyclic codes over Galois rings were investigated in \cite{CGFC2015}. In this paper, we investigate the same problem but over a more general ring family, finite commutative chain rings. When we focus on non-Galois finite commutative chain rings, we observe two different kinds of additivity. 

The first one, so-called Galois-additivity, is a natural generalization of the study in \cite{CGFC2015}, anyway our way of construction in this generalization is slightly different from the one in that paper. The authors in \cite{CGFC2015} were using some linear codes over the base ring and their generator matrices, but we just make use of ideals and do not get involved generator matrices. Our main result with this approach is Theorem \ref{thm:main1}. Also we have some further results related to the code size relations (Corollary \ref{cor:card}) and self-duality (Corollary \ref{cor:selfadd}) as well as we illustrate these results in a concrete example (Example \ref{ex:Gal}).   

The second one, so-called Eisenstein-additivity, is set up in Theorem \ref{thm:main2}. Eisenstein-additivity has some unusual properties especially while constructing dual codes. It is because, some chain rings are not free over their coefficient rings and hence we can not define a trace function for some elements over the base ring (Lemma \ref{lem:free}). Hence one can not make use of the equivalence of Euclidean orthogonality and duality via the trace map (see \cite[Lemma 6]{SW2004}). Thus we use the general idea of duality, constructed via annihilators of characters (see \cite{Woo2008}). We have adapted this character theoretic duality notion to Eisenstein-additive codes (see Section \ref{sec:EisDual}) and hence we observe one-to-oneness between a code and its dual (a MacWilliams identity) as expected. We again provide a concrete example for our result regarding the character theoretic duality (see Example \ref{ex:Eis}). 

Notice that the idea in \cite{CGFC2015} has been generalized recently also in \cite{SM2016}. However, the generalization in \cite{SM2016} is from Galois rings to free $R$-algebras, where $R$ is a finite commutative chain ring. Recall that we consider also non-free algebras (a finite commutative chain ring does not have to be a free module over a subring of it). On the other hand, our paper does not cover \cite{SM2016} since every free $R$-algebra does not have to be a chain ring.

\subsection{Organization of the paper}

In Section \ref{sec:fccr}, we introduce finite commutative chain rings constructing them step by step, as from $\Z_{p^n}$ to Galois rings and then to arbitrary finite commutative chain rings. Section \ref{sec:codes} provides basic definitions, notations and facts of a code concept over rings, by focusing mainly on cyclic codes and additive codes. 

In Section \ref{sec:addGal}, firstly we give some lemmas and then construct the main theorem (Theorem \ref{thm:main1}) of Galois-additivity using them. We also give some corollaries regarding the code size relations between dual codes (Corollary \ref{cor:card}), and characterization of self-duality (Corollary \ref{cor:selfadd}). We have also provided an illustration (Example \ref{ex:Gal}) about our results.

In Section \ref{sec:addEis}, we directly give the main result (Theorem \ref{thm:main2}) about Eisenstein-additive codes since it comes from the lemmas in Section \ref{sec:addGal}. A direct character theoretic duality is constructed in Section \ref{sec:EisDual} separately. Also an example (Example \ref{ex:Eis}) is available to illustrate our results.

\section{Finite Commutative Chain Rings}\label{sec:fccr}

A ring $R$ is called \textit{local} if it has only one maximal ideal. A local ring $R$ is called a \textit{chain ring} if its ideals form a chain under inclusion. Saying ``finite" we will refer to having finitely many elements (not being finitely generated). A finite chain ring is a principal ideal ring, and its maximal ideal is its nil-radical (i.e. the set of nil-potent elements). Hence, the chain of ideals of a finite chain ring $R$ is of the form
$$R\supsetneq NilRad(R)=<x> \supsetneq ...\supsetneq <x^{m-1}>\supsetneq <x^m>=<0>$$
for some idempotent $x\in R$ and a positive integer $m$ (see, for example \cite[Theorem 3.2]{Nec2008}). The simplest example of finite chain rings is $\Z_{p^n}$ with the maximal ideal $<p>$, where $p$ is a prime and $n$ is a positive integer. We may construct all finite chain rings using $\Z_{p^n}$.

\subsection{Galois Rings} 

Consider the ring homomorphism $\Z_{p^n}\rightarrow \F_p$ given by $a\mapsto \overline{a}$, where $\overline{a}$ is the remainder of $a$ modulo $p$ and $\F_q$ denotes the finite field of $q$ elements. We may extend this homomorphism in a natural way from the polynomial ring $\Z_{p^n}[X]$ to the polynomial ring $\F_p[X]$ by $\overline{(\sum_{i}a_iX^i)}=\sum_{i}\overline{a_i}X^i$. A polynomial over $\Z_{p^n}$ is called \textit{basic irreducible (primitive)} if its image under this homomorphism is irreducible (primitive) over $\F_p$. 

Let $f(X)\in\Z_{p^n}[X]$ be a basic irreducible polynomial of degree $r$. The quotient ring $\Z_{p^n}[X]/<f(X)>$ is a finite chain ring with the maximal ideal $<p>$. This kind of rings are known as \textit{Galois rings} and denoted by $GR(p^n,r)$. In a Galois ring $GR(p^n,r)$, $p^n$ is called \textit{characteristic} and $r$ is called \textit{rank}. Galois rings are unique up to isomorphism for a given characteristic and rank. We will use the notation $\Z_{p^n}[\omega]$ to denote Galois rings, taking $\omega=X+<f(X)>$. Note that $f(X)$ is the unique monic polynomial of degree less than or equal to $r$ such that $f(\omega)=0$. We may extend the ``overline homomorphism" defined above as from $\Z_{p^n}[\omega]$ to $\F_{p^r}$ such that $\overline{\omega}$ satisfies $\overline{f}(\overline{\omega})=0$ and so $\F_{p^r}=\F_{p}[\overline{\omega}]$. Therefore, in a Galois ring $GR(p^n,r)$, there exists an element of multiplicative order $p^r-1$, which is a root of a basic primitive polynomial of degree $r$ over $\Z_{p^n}$ and dividing $X^{p^r-1}-1$ in $\Z_{p^n}[X]$. If $\omega$ is a this kind of ``basic primitive element", then the set 
$$
T=\{0,1,\omega,\omega^2,...,\omega^{p^r-2}\}
$$
is called a \textit{Teichmuller set} of the Galois ring $GR(p^n,r)$. Any element $a\in GR(p^n,r)$ can be written uniquely as
\begin{equation}
a=a_0+a_1p+...+a_{n-1}p^{n-1},
\label{eq:GRelt}
\end{equation}
where $a_0,...,a_{n-1}\in T$. Here, $a$ is unit if and only if $a_0\neq 0$. Conversely, if $a$ is not unit, then it is either zero or a zero divisor. 

When we compare two Galois rings having the same characteristic, we observe that, $GR(p^n,r')$ is an extension of $GR(p^n,r)$ if and only if $r$ divides $r'$. Let $r'=rs$ where $s$ is a positive integer, the extension $GR(p^n,rs)$ of $GR(p^n,r)$ can be constructed by a basic primitive polynomial $g(X)\in GR(p^n,r)[X]=\Z_{p^n}[\omega][X]$ of degree $s$ with $\overline{g}(X)\in \F_{p^r}[X]=\F_{p}(\overline{\omega})[X]$. Taking $\zeta=X+<g(X)>$ in $GR(p^n,r)[X]/<g(X)>$, we can see that $\zeta$ has multiplicative order $p^{rs}-1$, $\omega =\zeta^{\frac{p^{rs}-1}{p^r-1}}$ and $\F_{p^{rs}}=\F_p(\overline{\zeta})$. For the proofs of these facts and more detailed information, see, for example \cite{Wan2003}.

\subsection{Finite Commutative Chain Rings}

We have constructed all Galois rings from $\Z_{p^n}$, and now we construct all finite commutative chain rings from Galois rings. The following theorem gives a complete characterization of finite commutative chain rings. For the proof and more detail, see, for example \cite[Theorem XVII.5]{McD1974}. 

\begin{theorem}\label{thm:fccr}
Consider the ring
\begin{equation}
R=\frac{GR(p^n,r)[x]=\Z_{p^n}[\omega][x]}{<g(x)=x^k+p(a_{k-1}x^{k-1}+...+a_0),p^{n-1}x^t>}, 
\label{eq:fccr-form}
\end{equation}
where $g(x)$ is an Eisenstein polynomial (i.e., $a_0$ is a unit element), $1\leq t \leq k$ when $n\geq 2$, and $t=k$ when $n=1$. $R$ is a finite commutative chain ring, and conversely, any finite commutative chain ring is of the form (\ref{eq:fccr-form}). 
\end{theorem}
\begin{notation}\label{not:fccr}
Since now, we fix $p,n,r,k,t,g$ and $x,\omega$ as the corresponding parameters and elements respectively of a finite commutative chain ring given in Theorem \ref{thm:fccr}.   
\end{notation}

The following proposition gives some elementary facts about finite commutative chain rings. For the proofs and more detail, see, for example \cite{McD1974,Nec2008}.

\begin{proposition} Let $R$ denote the finite commutative chain ring with the parameters in Notation \ref{not:fccr}.
\begin{itemize}
	\item The chain of ideals of $R$ is of the form $$R\supsetneq <x> \supsetneq ...\supsetneq <x^{k(n-1)+t-1}>\supsetneq <x^{k(n-1)+t}>=0.$$ Moreover, $<x^k>=<p>$.
	\item $|R|=p^{r(k(n-1)+t)}$ and $|<x>|=p^{r(k(n-1)+t-1)}$. Also, $R/<x>\cong \F_{p^r}$.
	\item The largest Galois ring in $R$ is $GR(p^n,r)$ and it is called the \textit{coefficient ring} of $R$.
	\item The Teichmuller set $T$ of the coefficient ring $S$ of a finite commutative chain ring $R$ is also considered as the Teichmuller set of $R$. Using $T$ and the generator $x$ of the maximal ideal, each element of $R$ can be written uniquely as
		$$
		a_0+a_1x+...+a_{k(n-1)+t-1}x^{k(n-1)+t-1},
		$$
for some $a_0,...,a_{r-1}\in T$. Let $T=\{0,1,\omega,...,\omega^{p^r-2}\}$, another unique representation of elements in $R$ using $\omega$ can be given as
		$$
		a_0+a_1\omega+...+a_{r-1}\omega^{r-1},
		$$
for some $a_0,...,a_{r-1}\in \Z_{p^n}[x]/<g(x),p^{n-1}x^t>$. Accordingly, we can give also the unique representation 
	$$
	\sum_{i=0}^{r-1}\sum_{j=0}^{k-1}a_{i,j}\omega^ix^j,
	$$
where $a_{i,j}\in\Z_{p^n}$ for $0\leq i\leq r-1$ and $0\leq j\leq k-1$, but $0\leq a_{i,j}< p^{n-1}$ when $t\leq j\leq k-1$. 
\end{itemize}
\end{proposition}

Now, we give an important idea we use in this paper. For the proof, see, for example \cite[Theorem 4.3.1]{BF2002} or \cite[Lemma 3.1]{DL2004}.  
\begin{theorem}\label{thm:ext}
Let $S=GR(p^n,r)$ and $S_e=GR(p^n,rs)=S[\zeta]$ for some basic primitive element $\zeta$ of degree $s>1$ over $S$. Let $R=S[x]/<g(x),p^{n-1}x^t>$ be a finite commutative chain ring constructed over $S$. Then $R_e=S_e[\zeta]$ is a finite commutative chain ring satisfying $R_e=S_e[x]/<g(x),p^{n-1}x^t>$. 
\end{theorem}
In other words, we can extend a finite chain ring using the extension of its coefficient ring. We can draw the tree of extensions as in Figure \ref{tree:fccr}.

\begin{figure}[h]
\caption{Tree of extensions of a finite commutative chain ring}
\centering
\begin{tikzpicture}[node distance=2cm]
\title{Tree of extensions of finite commutative chain rings}
\label{tree:fccr}
\node(Se)							                      {$R_e=\frac{S_e[x]}{<g(x),p^{n-1}x^t>}=R[\zeta]$};
\node(S)   	[below right=3cm and 1.5cm of Se] {$R=\frac{S[x]}{<g(x),p^{n-1}x^t>}$};
\node(Re)   [below left=3cm and 4.5cm of Se] 	{$S_e=S[\zeta]$};
\node(R)   	[below left=6cm and 3cm of Se] {$S=GR(p^n,r)$};

\node(d1)   [below right=1cm and 0.7cm of Se] {$s$ (by $\zeta$)};
\node(d2)   [below right=4.8cm and -1.5cm of Se] {$k-\frac{k-t}{n}$ (by $g(x)$ and $t$)};
\node(d3)   [below left=0.6cm and 1.8cm of Se] {$k-\frac{k-t}{n}$};
\node(d31)  [below left=1.1cm and 1.7cm of Se] {(by $g(x)$ and $t$)};
\node(d4)   [below left=4.1cm and 2.2cm of Se] {$s$ (by $\zeta$)};

\draw(Se)--(S);
\draw(Se)--(Re);
\draw(R)--(S);
\draw(R)--(Re);
\end{tikzpicture}
\end{figure}
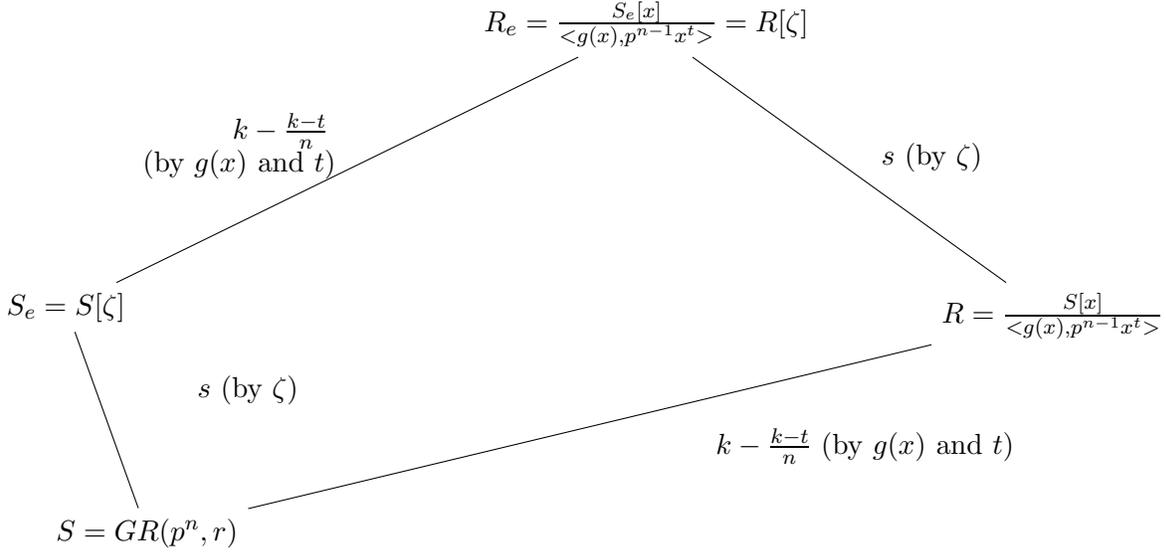
\begin{remark}
Let $S$ be a finite commutative chain ring and $R$ be an extension of it. It makes sense to determine a kind of ``degree of extension" considering $Log_{|S|}(|R|)$. Remark that this degree of extension may not be an integer. See Figure \ref{tree:fccr} and as an example see Figure \ref{tree:fccrEx}. 
\end{remark}

\begin{remark} 
Notice that a finite commutative chain ring is a Galois ring if and only if $k=1$. In case $n=1$, a finite commutative chain ring is of the form $\F_{p^r}[x]/<x^k>$ and it is called \textit{quasi-Galois ring}. Similar to Galois rings, quasi-Galois rings are unique up to isomorphism for a given characteristic and cardinality. However, a proper (neither Galois nor quasi-Galois) finite commutative chain ring may not be unique up to isomorphism for given characteristic, rank and cardinality. Example \ref{ex:1}(ii) below illustrates this fact. 
\end{remark}

\begin{example}\label{ex:1}
\begin{itemize}
	\item[i)] Consider 
$$R=\frac{\Z_4[x]}{<x^2+2,2x>}\ \left(=\frac{\Z_4[x]}{<x^2+2x+2,2x>}\right).$$ 
$R$ has $8$ elements and is additively equivalent to $\Z_2\oplus\Z_4$. It has $4$ units and its multiplicative group is isomorphic to $\Z_4$. Also $R/<x>\cong \F_2$.  
	\item[ii)] Consider
$$R_1=\frac{\Z_4[x]}{<x^2+2,2x^2>}\ \left(=\frac{\Z_4[x]}{<x^2+2>}\right).$$
$R_1$ has $16$ elements, and additively equivalent to $\Z_4\oplus\Z_4$. Similarly,
$$R_2=\frac{\Z_4[x]}{<x^2+2x+2,2x^2>}\ \left(=\frac{\Z_4[x]}{<x^2+2x+2>}\right)$$
has $16$ elements and additively equivalent to $\Z_4\oplus\Z_4$. Moreover, $R_1^*\cong R_2^*\cong \Z_2\oplus \Z_4$, where $R_i^*$ denotes the set of units for $1\leq i\leq 2$. Also $R_1/<x>\cong R_2/<x> \cong \F_2$. However, $R_1$ and $R_2$ are not isomorphic as rings.  
	\item[iii)] Consider
$$R_e=\frac{\Z_4[\zeta][x]}{<x^2+2,2x>},$$
where $\zeta$ is a root of the basic primitive polynomial $f(X)=X^2+X+1\in\Z_4[X]$. $R_e$ has $64$ elements, and is additively equivalent to $\Z_2\oplus\Z_2\oplus\Z_4\oplus\Z_4$. Moreover, $R_e$ is an extension of $R$ given in (i) with $R_e=R[\zeta]$, and also $R_e/<x>\cong \F_4=\F_2(\overline{\zeta})$. Similarly, $S_e=\Z_4[\zeta]$ is a Galois ring as an extension of $S=\Z_{4}$. See Figure \ref{tree:fccrEx} for the tree of extension (remark that the extension degree of $R_e$ over $S_e$ is not an integer).   
\end{itemize}
\end{example}

\begin{figure}[h]
\caption{Tree of extensions: a proper example}
\centering
\begin{tikzpicture}[node distance=2cm]
\title{Tree of extensions in Example \ref{ex:1}(iii)}
\label{tree:fccrEx}
\node(Se)							                      {$R_e=\frac{\Z_4[\zeta][x]}{<x^2+2,2x>}$};
\node(S)   	[below right=3cm and 1.5cm of Se] {$R=\frac{\Z_4[x]}{<x^2+2,2x>}$};
\node(Re)   [below left=3cm and 4.5cm of Se] 	{$S_e=\Z_4[\zeta]$};
\node(R)   	[below left=6cm and 3cm of Se] {$S=\Z_4$};

\node(d1)   [below right=1cm and 0.7cm of Se] {$2$};
\node(d2)   [below right=4.8cm and -1.5cm of Se] {$\frac{3}{2}$};
\node(d3)   [below left=0.6cm and 1.8cm of Se] {$\frac{3}{2}$};
\node(d4)   [below left=4.1cm and 3.7cm of Se] {$2$};

\draw(Se)--(S);
\draw(Se)--(Re);
\draw(R)--(S);
\draw(R)--(Re);
\end{tikzpicture}
\end{figure}
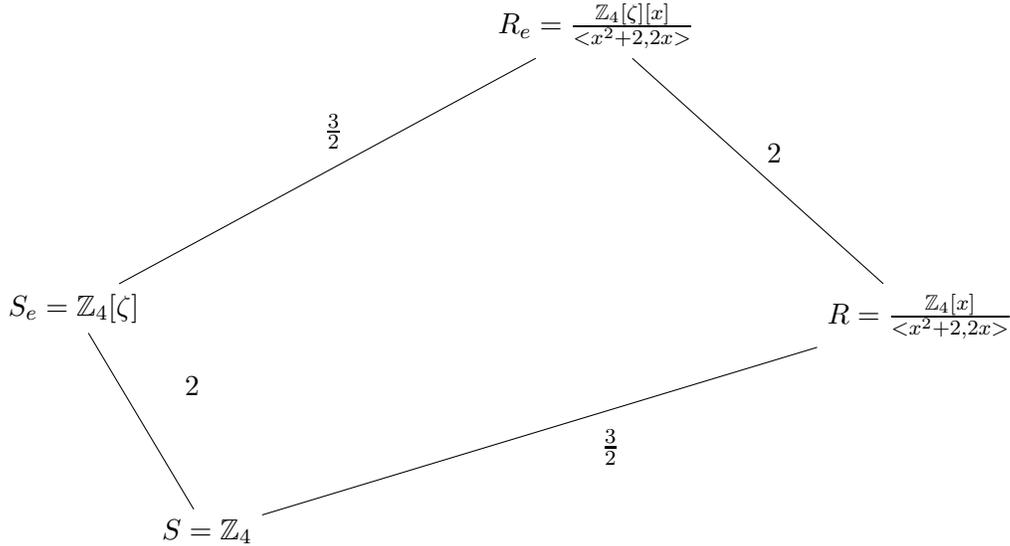

Polynomials over finite commutative finite chain rings have a kind of ``unique factorization" under some assumptions. The following theorem, which is also known as \textit{Hensel's lemma}, demonstrates this property. It is the adapted version of Hensel's result in \cite{Hen1918} to our context.

\begin{theorem}\label{thm:hen}
Let $R$ be a finite commutative chain ring and $I$ be its maximal ideal. Also let $h(X),h_1(X), ..., h_l(X)\in R[X]$ be monic polynomials such that 
\begin{itemize}
	\item $h(X)= h_1(X)\cdots h_l(X)$,
	\item $\overline{h_i}(X)\in R[X]/I$ are pairwise coprime for all $1\leq i\leq l$, and
	\item $\overline{h_i}(0)\neq 0$ for all $1\leq i\leq l$.
\end{itemize}
Then, the factorization $h(X)= h_1(X)\cdots h_l(X)$ is unique up to reordering. 
\end{theorem}
The polynomials $h_i$ in Theorem \ref{thm:hen} are known as the \textit{Hensel's lifting} of the polynomials $\overline{h_i}$, for all $1\leq i\leq l$. 

Now we give automorphisms of finite chain rings. Let 
$$
R=\frac{GR(p^n,r)[x]}{<g(x),p^{n-1}x^t>} \text{ and } R_e=\frac{GR(p^n,rs)[x]}{<g(x),p^{n-1}x^t>}
$$
be two finite commutative chain rings. Hence, $R_e$ is an extension of $R$. Let $R_e=R[\zeta]$ for some basic primitive element $\zeta\in GR(p^n,rs)$ over $GR(p^n,r)$ of degree $s$, the automorphisms of $R_e$ keeping the elements in $R$ fixed are of the form
\begin{equation}
\label{eq:aut}
a_0+a_1\zeta+...+a_{s-1}\zeta^{s-1}\mapsto a_0+a_1\zeta^{q^{i}}+...+a_{s-1}\zeta^{(s-1)q^{i}}
\end{equation}
for some $0\leq i\leq s-1$, where $q=p^r$ and $a_j\in R$ for all $0\leq j\leq s-1$ (for the proofs and more detail, see, for example \cite[Corollary 5.1.5]{BF2002}). We will denote the automorphism group (i.e., the set of such maps with composition) by $Aut_R(R_e)$. Remark that the roots of a basic irreducible polynomial over $R_e$ are conjugate each other under this automorphism.

\section{Codes over finite commutative chain rings}\label{sec:codes}

Let $R$ be a ring and $N$ be a positive integer, a non-empty subset $\sC$ of $R^N=\{(a_1,...,a_N):a_i\in R\text{ for all }i=1,..,N\}$ is called a \textit{code} over $R$ of length $N$. An element of a code $\sC$ is called \textit{codeword}. The function $d$ on $R^N\times R^N$ given by $d(a,b)=|\{i:a_i\neq b_i\}|$ for all $a=(a_1,...,a_N),b=(b_1,...,b_N)\in R^N$ is a metric and it is called \textit{Hamming distance}. \textit{Minimum distance} $d(\sC)$ of a code $\sC$ is defined by $d(\sC)=\min\{d(a,b):a,b\in\sC\text{ and }a\neq b\}$. The \textit{weight} of a codeword $a\in \sC$ is defined by $w(a)=d(a,0)$. Similarly, the \textit{minimum weight} of a code $\sC$ is defined as the minimum non-zero weight in $\sC$ and denoted by $w(\sC)$. 

Let $S$ be a subring of $R$, then we say that a code $\sC$ is $S$-\textit{linear} if it is a module over $S$. Such codes are also known as \textit{additive} codes, and in particular, $R$-linear codes are known as \textit{linear} codes. Note that, if $\sC$ is an additive code, then $d(\sC)=w(\sC)$. This property and some other properties make additive codes quite interesting and useful in both theory and application. 

The \textit{Euclidean inner product} of $R^N$ is defined by $(a ,b)_E=\sum_{i=1}^{N}a_ib_i$ for all $a=(a_1,...,a_N)$ and $b=(b_1,...,b_N)$ in $R^N$. When a code $\sC\subseteq R^N$ is given, the \textit{Euclidean orthogonal code} $\sC^{\perp_E}$ of $\sC$ is given by $\sC^{\perp_E}=\{a\in R^N:(a,b)_E=0\text{ for all }b\in \sC\}$. It can be shown directly that the Euclidean orthogonal code of a code is a linear code and that in the case of linear codes over Galois rings it coincides with the dual code, see \cite{SW2004}.
  
A code $\sC$ is \textit{cyclic} if $(a_N,a_1,...,a_{N-1})\in\sC$ for all $(a_1,...,a_{N})\in\sC$. Cyclic codes correspond to a subset of $R[X]/<X^N-1>$, which is closed under multiplication by $X$. In particular, there is a one to one correspondence between a linear cyclic code and an ideal of $R[X]/<X^N-1>$. In the literature, the term ``cyclic code" is generally understood as ``linear cyclic code". However, there are non-linear cyclic code constructions, too. A recent example is \cite{CGFC2015}, in which the authors investigated additive cyclic codes over Galois rings. In detail, the authors constructed $\Z_{p^n}$-linear cyclic codes over $\Z_{p^n}[\omega]$, where $\omega$ is a basic primitive element of degree prime $r$. They also investigated the duality of such codes. In the following two sections, we generalize their results from Galois rings to finite commutative chain rings. 

In Section \ref{sec:addGal}, we construct $S=\Z_{p^n}[x]/<g(x),p^{n-1}x^t>$-linear cyclic codes over $R=\Z_{p^n}[\omega][x]/<g(x),p^{n-1}x^t>$ and investigate the duality of such codes, where $g(x)\in\Z_{p^n}[x]$ is a monic Eisenstein polynomial of degree $k$ and $\omega$ is a basic primitive element over $\Z_{p^n}$ of degree prime $r$. Our results in this set up occur as a natural generalization of the results in \cite{CGFC2015}. However, the method we use in our generalization is slightly different from the one in \cite{CGFC2015}. We just make use of ideals and do not get involved generator matrices, whereas the authors in \cite{CGFC2015} were using some linear codes over the base ring $S$ and their generator matrices. We call such codes \textit{Galois-additive codes} because $R$ is a Galois extension of $S$.

In Section \ref{sec:addEis}, we construct $S=\Z_{p^n}$-linear cyclic codes over $R=\Z_{p^n}[x]/<g(x),p^{n-1}x^t>$ and investigate the duality of such codes, where $g(x)\in\Z_{p^n}[x]$ is a monic Eisenstein polynomial of degree $k$. Our results in this set up are not very similar to the ones in Section \ref{sec:addGal}, especially considering duality. The main reason here is the fact that $R$ is not a free module over $S$. Therefore, we can not ensure that we can generate all the additive characters by means of trace functions. Hence, in the case we use Euclidean inner product (Euclidean orthogonality) we observe that some properties are not satisfied as expected. Thus we use directly the 
duality notion in Section 5.1 as stated in \cite{Woo2008} using character theory. Since $R$ is obtained by an Eisenstein polynomial from $S$, we call such additive codes as \textit{Eisenstein-additive codes}.      

\section{Galois-additive cyclic codes}\label{sec:addGal}
	
Let $R=\frac{\Z_{p^n}[\omega][x]}{<g(x),p^{n-1}x^t>}$ and $S=\frac{\Z_{p^n}[x]}{<g(x),p^{n-1}x^t>}$, where $p$ is a prime, $n$ is a positive integer, $\omega$ is a basic irreducible element over $\Z_{p^n}$ of degree $r$ (which is a prime) and $g(x)\in\Z_{p^n}[x]$ is an Eisenstein polynomial of degree $k$. Also we denote the nilpotency index by $m=k(n-1)+t$. Clearly $S$ is a subring of $R$. Let $\sR=R[X]/<X^N-1>$ and $\sS=S[X]/<X^N-1>$, where $N$ is a positive integer satisfying $\gcd (N,p)=1$. Similarly $\sS$ is a subring of $\sR$. On the other hand, $\sR$ is a module over $R$ (and also over $S$). In addition, there is a one to one correspondence between $\sR$ and $R^N$ given by $a_0+a_1X+...+a_{N-1}X^{N-1}\leftrightarrow (a_0,a_1,...,a_{N-1})$. On the grounds of this correspondence, we will say ``code" for both a non-empty subset of $\sR$ and a non-empty subset of $R^N$. In this section, we will construct $S$-linear cyclic codes $\sC\subseteq \sR$ and investigate their duality. 
	
Let $C_i^{(b)}=\{ib^j \mod N: j\in\Z\}$, where $b\in\{p,p^r\}$. Also let $\kappa_i=|C_i^{(p)}|$ and $\kappa_{i,j}=|C_{ip^j}^{(p^r)}|$, for $0\leq j\leq r-1$. Then we have the following.
	\begin{enumerate}
		\item If $\gcd (\kappa_i,r)=1$, then $C_i^{(p)}=C_i^{(p^r)}$.
		\item If $r$ divides $\kappa_i$, then $\kappa_i=r|C_i^{(p^r)}|$ and $C_i^{(p)}=\bigcup_{j=0}^{r-1}C_{ip^j}^{(p^r)}$.
	\end{enumerate}

Let $s=\min \{j\in\Z^+:(p^r)^j\equiv 1 \mod N\}$ and $\widehat{R}=R[\zeta]$, where $\zeta$ is a basic primitive element over $R$ of degree $s$. Notice that the multiplicative order of $\zeta$ is $p^{rs}-1$ and $\omega =\zeta^{\frac{p^{rs}-1}{p^{r}-1}}$. The set $\widehat{T}=\{0,1,\zeta,...,\zeta^{p^{rs}-2}\}$ is a Teichmuller set of $\widehat{R}$. For each element $a\in\widehat{R}$, consider the unique representation $a=a_0+a_1x+...+a_{m-1}x^{m-1}$ of $a$, where $a_i\in\widehat{T}$ for $0\leq i \leq m-1$. Using this representation, let $\widehat{\phi}:a\mapsto a_0^p+a_1^px+...+a_{m-1}^px^{m-1}$. Notice that $\widehat{\phi}(b)=b$ if and only if $b\in R$. $\widehat{\phi}$ can be extended as $\widehat{R}[X]\rightarrow \widehat{R}[X]$ by $\sum_{i}a_iX^i\mapsto \sum_{i}\widehat{\phi}(a_i)X^i$. 

Similarly, the set $T=\{0,1,\omega,...,\omega^{p^{r}-2}\}$ is a Teichmuller set of $R$. Also we may define $\phi=\widehat{\phi}|_R$.  
	
Let $\eta =\zeta^{\frac{p^{rs}-1}{N}}$, clearly it is a primitive $N.$th root of unity. Also let $C_{j_0}^{(p)},C_{j_1}^{(p)},...,C_{j_v}^{(p)}$ be all distinct $p-$cyclotomic cosets modulo $N$, where $j_0=0$ and $1\leq j_1<...<j_v\leq N-1$ such that $\gcd (\kappa_{j_i},r)=1$ for $0\leq i\leq u$, and $\kappa_{j_i}$ is a multiple of $r$ for $u+1\leq i\leq v$. Hence, $C_{j_i}^{(p)}=C_{j_i}^{(p^r)}$ for $0\leq i\leq u$, $C_{j_i}^{(p)}\neq C_{j_i}^{(p^r)}$ and $C_{j_i}^{(p)}=C_{j_i}^{(p^r)}\cup C_{j_ip}^{(p^r)}\cup ...\cup C_{j_ip^{r-1}}^{(p^r)}$ for $u+1\leq i\leq v$.
	
Let $m_i(X)=\prod_{j\in C_{j_i}^{(p)}}X-\eta^j$ and $R_i=S[X]/<m_i(X)>$ for $0\leq i\leq v$. $m_i(X)$ is a basic irreducible polynomial over $S$. Now, let
$$
\epsilon_i(X)=\frac{1}{N}\sum^{N-1}_{j=0}\sum_{l\in C_{j_i}^{(p)}}\eta^{-jl}X^j\in \widehat{R}[X],
$$
for $0\leq i\leq v$. Here, $\widehat{\phi}(\epsilon_i(X))=\epsilon_i(X)$ and so $\epsilon_i(X)\in S[X]$, for all $0\leq i\leq v$. Let also $\sK_i=\epsilon_i (X)\sS$, for $0\leq i\leq v$. The following lemma can be proved by straightforward computations. 

\begin{lemma}\label{epsilon_i}
The following hold for all $0\leq i,j\leq v$.
	\begin{enumerate}
		\item $\epsilon_i(X)^2=\epsilon_i(X)$ and $\epsilon_i(X)\epsilon_j(X)= 0$ in $\sS$ when $i\neq j$. Also $\epsilon_0(X)+\epsilon_1(X)+...+\epsilon_v(X)=1$ in $\sS$.
		\item $\sS=\sK_0+\sK_1+...+\sK_v=\sK_0\oplus\sK_1\oplus ...\oplus\sK_v$, and $\epsilon_i(X)$ is the multiplicative identity of $\sK_i$.
		\item The map $\psi_i:R_i\rightarrow \sK_i$ given by $f(X)+<m_i(X)>\mapsto \epsilon_i(X)f(X)+<X^N-1>$ is a ring isomorphism. Hence, $\sK_i$ has a basis $\{\epsilon_i(X)X^j:j=0,1,...,\kappa_{j_i}-1\}$ over $S$.  
	\end{enumerate}
\end{lemma}
	
Now let $m_{i,h}=\prod_{j\in C_{j_ip^h}^{(p^r)}}X-\eta^j$ for $0\leq h\leq r-1$ and $u+1\leq i\leq v$. Similarly, $m_{i,h}(X)$ is a basic irreducible polynomial over $R$. Also let 
$$
\epsilon_{i,h}(X)=\frac{1}{N}\sum^{N-1}_{j=0}\sum_{l\in C_{j_ip^h}^{(p^r)}}\eta^{-jl}X^j\in \widehat{R}[X],
$$ 
for $u+1\leq i\leq v$ and $0\leq h\leq r-1$. Here, $\widehat{\phi}^r(\epsilon_{i,h}(X))=\epsilon_{i,h}(X)$ and so $\epsilon_{i,h}(X)\in R[X]$, for all $u+1\leq i\leq v$ and $0\leq h\leq r-1$. Let also $\sL_{i}=\epsilon_{i} (X)\sR$ and $\sL_{i,h}=\epsilon_{i,h} (X)\sR$, the following lemma, similar to Lemma \ref{epsilon_i}, can be proved directly by computations. 

\begin{lemma}\label{epsilon_ij} The following hold for all $u+1\leq i\leq v$ and $0\leq h,j\leq r-1$.
	\begin{enumerate}
		\item $\epsilon_{i,h}(X)^2=\epsilon_{i,h}(X)$ and $\epsilon_{i,j}(X)\epsilon_{i,h}(X)=0$ in $\sR$, when $j\neq h$. Also $\epsilon_{i,0}(X)+\epsilon_{i,1}(X)+...+\epsilon_{i,r-1}(X)=\epsilon_{i}(X)$ in $\sR$.
		\item $\sR=\sL_{0}+\sL_1+...+\sL_v=\sL_{0}\oplus\sL_1\oplus ...\oplus\sL_v$, and $\epsilon_i(X)$ is the multiplicative identity of $\sL_i$.
		\item $\sL_i=\sL_{i,0}+\sL_{i,1}+...+\sL_{i,r-1}=\sL_{i,0}\oplus\sL_{i,1}\oplus ...\oplus\sL_{i,r-1}$, and $\epsilon_{i,h}(X)$ is the multiplicative identity of $\sL_{i,h}$. 
		\item The map $\psi_i:R[X]/<m_i(X)>\rightarrow \sL_i$ given by $f(X)+<m_i(X)>\mapsto \epsilon_i(X)f(X)+<X^N-1>$ is a ring isomorphism and hence $\sL_i$ has a basis $\{\epsilon_i(X)X^j:j=0,1,...,\kappa_{j_i}-1\}$. Similarly, the map $\psi_{i,h}:R[X]/<m_{i,h}(X)>\rightarrow \sL_{i,h}$ given by $f(X)+<m_{i,h}(X)>\mapsto \epsilon_{i,h}(X)f(X)+<X^N-1>$ is a ring isomorphism and hence $\sL_{i,h}$ has a basis $\{\epsilon_{i,h}(X)X^j:j=0,1,...,\kappa_{j_i,h}-1\}$. 
	\end{enumerate}
\end{lemma}	

\begin{corollary}\label{cor:4.1}
In $\sR$, $\epsilon_i(X)\epsilon_{j,h}(X)=0$ when $i\neq j$, for all $0\leq i\leq v$, $u+1\leq j\leq v$ and $0\leq h\leq r-1$.
\end{corollary}
\begin{proof}
For any arbitrary $0\leq i\leq v$, $u+1\leq j\leq v$ and $0\leq h\leq r-1$, Lemma \ref{epsilon_i}(1) and Lemma \ref{epsilon_ij}(1) implies $$\epsilon_i(X)\epsilon_{j,h}(X)=-\sum_{l\neq j}\epsilon_i(X)\epsilon_{l,h}.$$ Multiplying both sides by $\epsilon_{j,h}(X)$, we obtain the statement of the corollary, by Lemma \ref{epsilon_ij}(1). 
\end{proof}
	
Clearly, $\sK_i =\sL_i\cap \sS$ for all $0\leq i\leq v$. However, when we define $\sK_{i,j} =\epsilon_{i,j}(X)\sS$, we obtain $\sK_{i,j} \neq \sL_{i,j}\cap \sS$, for all $u+1\leq i\leq v$ and $0\leq j\leq r-1$. The following lemma expresses the idea behind it in a precise way.   

\begin{lemma}\label{lem:sK-sL}
$\sK_{i,j}=\sL_{i,j}$ for all $u+1\leq i\leq v$ and $0\leq j\leq r-1$.
\end{lemma}
\begin{proof}
Firstly notice that
\begin{equation}
\begin{array}{lll}
\sL_{i,h} & = & \epsilon_{i,h}(X)\sR \\
\ & = & 				\epsilon_{i,h}(X)\sum_{j=0}^{r-1}\omega^j\sS \\
\ & = & 				\sum_{j=0}^{r-1}\omega^j\epsilon_{i,h}(X)\sS \\
\ & = & 				\sum_{j=0}^{r-1}\omega^j\sK_{i,h},
\end{array}
\label{eq:sK-sL}
\end{equation}
for all $u+1\leq i\leq v$ and $0\leq j\leq r-1$. Equation (\ref{eq:sK-sL}) says that, $\sK_{i,j}=\sL_{i,j}$ if and only if $\sK_{i,h}=\omega\sK_{i,h}$. Hence, we will equivalently prove $\sK_{i,h}=\omega\sK_{i,h}$,  for all $u+1\leq i\leq v$ and $0\leq h\leq r-1$. 

Now, assume the contrary, i.e., assume $\sK_{i_0,h_0}\neq\omega\sK_{i_0,h_0}$ for some $u+1\leq i_0\leq v$ and $0\leq h_0\leq r-1$. Then, since $\epsilon_{i,j}(X)\neq\omega^l\epsilon_{i,h}(X)$ when $j\neq h$ and $1\leq l\leq r-1$, the equation (\ref{eq:sK-sL}) says that $\sL_{i_0}=\epsilon_{i_0}(X)\sR$ has a basis over $\sK_{i_0}=\epsilon_{i_0}(X)\sS $ including more than $r$ elements, this is a contradiction. Then, $\sK_{i,h}=\omega\sK_{i,h}$, for all $u+1\leq i\leq v$ and $0\leq j\leq r-1$.   
\end{proof}

It can be shown directly that $\phi$ satisfies the following, for all $0\leq i\leq v$ and $0\leq h\leq r-1$.
	\begin{enumerate}
		\item $\phi (\epsilon_i (X))=\epsilon_i (X)$ and $\phi (\epsilon_{i,h} (X))=\epsilon_{i,h+1\mod r} (X)$.
		\item $\phi (\sL_i)=\sL_i$ and $\phi (\sL_{i,h})=\sL_{i,h+1\mod r}$.
	\end{enumerate}

The following lemma is immediate.
\begin{lemma}\label{lem:cyclicness}	
Let $\sC$ be a non-empty subset of $\sR$. $\sC$ is an $S$-linear code over $R$ of length $N$ if and only if there is a unique $\sK_i$-submodule $\sC_i$ of $\sL_i$ such that $\sC =\bigoplus_{i=0}^{v}\sC_i$. Hence we have also $|\sC|=\prod_{i=0}^v|\sC_i|$. 
\end{lemma}

Let $\sC$ be a code and $\sC_i$ be given as in Lemma \ref{lem:cyclicness}. The unique decomposition $\sC =\bigoplus_{i=0}^{v}\sC_i$ is called \textit{canonical decomposition} of $\sC$.  	
	
	
\textit{Euclidean inner product} of $a (X) =a_0+a_1X+...+a_{N-1}X^{N-1}$ and $b (X) =b_0+b_1X+...+b_{N-1}X^{N-1}$ in $\sR$ is naturally given by $(a(X),b(X))_E=\sum_{i=0}^{N-1}a_ib_i$. Accordingly, \textit{Euclidean orthogonal code} $\sC^{\perp_E}$ of a code $\sC$ is given by $\sC^{\perp_E}=\{a\in \sR:(a (X),b (X))_E=0\text{ for all }b (X)\in \sC\}$. 

Let $Tr:R\rightarrow S$ be given by $c\mapsto \sum_{j=0}^{r-1}\phi^j(c)$. $Tr$ is called the\textit{ generalized trace} of $R$ relative to $S$. \textit{Trace inner product} of $\sR$ over $S$ is defined by $(a(X),b(X))_{Tr}=\sum_{j=0}^{N-1}Tr(a_ib_i)$ for $a (X) =a_0+a_1X+...+a_{N-1}X^{N-1},b (X) =b_0+b_1X+...+b_{N-1}X^{N-1}\in \sR$. Accordingly, \textit{Trace orthogonal code} $\sC^{\perp}$ of a code $\sC$ is given by $\sC^{\perp}=\{a\in \sR:(a (X),b(X))_{Tr}=0\text{ for all }b (X)\in \sC\}$.   
	
Let $\mu:\sR\rightarrow \sR$ be given by $a(X)\mapsto X^Na(X^{-1})$ for all $a(X)\in \sR$. Clearly, $\mu$ is an automorphism on $\sR$ of order 2. Also $\mu$ can be used as the permutation on $\{0,1,...,v\}$ given by $i\mapsto i'$ such that $C_{j_{i'}}=C_{-j_{i}}$. The following lemma is straightforward. 

\begin{lemma} The map $\mu$ defined in the previous paragraph has the following properties.
	\begin{enumerate}
		\item $\mu (\epsilon_i(X))=\epsilon_{\mu (i)}(X)$ in $\sR$ for all $0\leq i\leq v$.
		\item $\mu (0)=0$, $1\leq \mu (i)\leq u$ for $1\leq i\leq u$, and $u+1\leq \mu (i)\leq v$ for $u+1\leq i\leq v$.
	\end{enumerate}
In addition, $\mu$ also determines a ring isomorphism $R_i\rightarrow R_{\mu (i)}$ given by $f(X)\mapsto \mu (f(X))\mod m_{\mu (i)}(X)$. Hence, $\mu (\sK_i)=\sK_{\mu (i)}$ and $\mu (\sL_i)=\sL_{\mu (i)}$ for all $0\leq i\leq v$.
\end{lemma}

The following lemma is one of the main arguments we use to prove our main theorem in this section. 
\begin{lemma}\label{lem:mu}
In $\sR$, $a(X)\mu (b(X))=0$ if and only if $(a(X),X^hb(X))_{E}=0$ for all $0\leq h\leq N-1$. 
\end{lemma}

\begin{proof} Let $a(X)=\sum_{i=0}^{N-1}a_iX^i$ and $b(X)=\sum_{i=0}^{N-1}b_iX^i$. Then
$$
\begin{array}{lll}
a(X)\mu (b(X)) & = & \sum_{i=0}^{N-1}\sum_{j=0}^{N-1} a_ib_jX^{i-j\mod N} \\	
\ & = & \sum_{i=0}^{N-1}\sum_{h=0}^{N-1} a_ib_{i-h\mod N}X^{h} \\	
\ & = & \sum_{h=0}^{N-1} \left(\sum_{i=0}^{N-1} a_ib_{i-h\mod N} \right)X^{h} \\	
\ & = & \sum_{h=0}^{N-1} \left(a(X),X^hb(X)\right)_EX^{h} .	
\end{array}
$$
\end{proof}

Recall that $\{1,\omega,...,\omega^{r-1}\}$ is an $S$-basis of $R$. Let  $\{\theta_0,\theta_1,...,\theta_{r-1}\}$ be another $S$-basis of $R$. These two basis are called \textit{trace dual} of each other if $T(\omega^j\theta_j)=1$ and $T(\omega^j\theta_h)=0$ for all $0\leq j\neq h\leq r-1$. When a basis is given, its trace dual can be constructed by the following lemma.

\begin{lemma}
Let $\gamma'(X)=\frac{\gamma (X)}{X-\omega}=\sum_{j=0}^{r-1}\gamma_jX^j\in R[X]$, where $\gamma(X)$ is the primitive irreducible polynomial of $\omega$ over $S$. Let also $\theta_j=\frac{\gamma_j}{\gamma'(\omega)}\in R$ for $0\leq j\leq r-1$. Then $\{\theta_0,...,\theta_{r-1}\}$ is the trace dual of $\{1,\omega,...,\omega^{r-1}\}$.
\end{lemma}

\begin{proof}
It can be proved directly by the interpolation idea of polynomials.  
\end{proof}

Now we give one of the main results in this paper. Note that, as we are in a Galois extension, trace orthogonality can be translated one to one to duality \cite[Lemma 6]{SW2004} thus we will use the term of trace duality.
\begin{theorem}\label{thm:main1}
Consider the definitions and notations given above. Any $S$-linear cyclic code $\sC\subseteq\sR$ over $R$ of length $N$ is of the form
\begin{equation}
\sC =\sum_{i=0}^{u}\sum_{j=0}^{r-1}x^{e_{i,j}}\omega^j\sK_{i} 
+\sum_{i=u+1}^{v}\sum_{j=0}^{r-1} x^{e_{i,j}}\sK_{i,j},
\label{eq:sC}
\end{equation}
for some $0\leq e_{i,j}\leq m-1$ given for $0\leq i\leq v$ and $0\leq j\leq r-1$. 

On the other hand, the trace dual code of $\sC$ given in (\ref{eq:sC}) is of the form 
\begin{equation}
\sC^{\perp} =\sum_{i=0}^{u}\sum_{j=0}^{r-1}x^{m-e_{i,j}}\theta_j\sK_{i} 
+\sum_{i=u+1}^{v}\sum_{j=0}^{r-1} x^{m-e_{i,j}}\sK_{i,j},
\label{eq:sCperp}
\end{equation}
where $Tr(\omega^i\theta_i)=1$ and $Tr(\omega^i\theta_j)=0$ for all $0\leq i\neq j\leq r-1$. 
\end{theorem}

\begin{proof}
By Lemmas \ref{epsilon_i}, \ref{epsilon_ij}, \ref{cor:4.1} and \ref{lem:sK-sL}, the canonical decomposition of $\sR$ is given by 
$$ 
\sR =\sum_{i=0}^{u}\sum_{j=0}^{r-1}\omega^j\sK_{i} +\sum_{i=u+1}^{v}\sum_{j=0}^{r-1}\sK_{i,j}
 =\bigoplus_{i=0}^{u}\bigoplus_{j=0}^{r-1}\omega^j\sK_{i} \oplus\bigoplus_{i=u+1}^{v}\bigoplus_{j=0}^{r-1}\sK_{i,j}. 
$$  
Notice that any $S$-linear subset of $\omega^j\sK_i$ (or $\sK_{i,j}$) is an ideal $x^{e_{i,j}}\omega^j\sK_i$ (or $x^{e_{i,j}}\sK_{i,j}$) of itself, for some $0\leq e_{i,j}\leq m-1$ given for $0\leq i\leq v$ and $0\leq j\leq r-1$. Therefore, the construction of $\sC$ is clear by Lemma \ref{lem:cyclicness}. 

Now, let 
$$
\sD =\sum_{i=0}^{u}\sum_{j=0}^{r-1}x^{m-e_{i,j}}\theta_j\sK_{i} 
+\sum_{i=u+1}^{v}\sum_{j=0}^{r-1}x^{m-e_{i,j}}\sK_{i,j},
$$ 
where $Tr(\omega^i\theta_i)=1$ and $Tr(\omega^i\theta_j)=0$ for all $0\leq i\neq j\leq r-1$. Also let $a(X)\in\sC$ and $b(X)\in\sD$ be arbitrary. From the properties Lemma \ref{epsilon_i}(1), Lemma \ref{epsilon_ij}(1), Corollary \ref{cor:4.1} and the trace duality between $\{1,\omega,...,\omega^{r-1}\}$ and $\{\theta_0,\theta_1,...,\theta_{r-1}\}$, we deduce 
$$
\sum_{l=0}^{r-1}\phi^l (a(X)b(X))=0, 
$$ 
which implies $(a(X),b(X))_{Tr}=0$ by Lemma \ref{lem:mu}. Therefore, $\sD\subseteq\sC^{\perp}$. Now, let $b'(X)\in\sR$ but $b'(X)\notin\sD$. Then similar arguments above can be used to show $(a(X),b'(X))_{Tr}\neq 0$ for some $a(X)\in\sC$. Therefore, $\sC^{\perp}\subseteq\sD$. Conclusively, $\sC^{\perp}=\sD$.  
\end{proof}

\begin{remark}
The special case of Theorem \ref{thm:main1} for $t=k=1$ corresponds to the construction in \cite{CGFC2015}. However, we use a different language without mentioning generator matrices of some subcodes over $S$. 
\end{remark}

As expected for a dual code we can also deduce the following results from Theorem \ref{thm:main1} related to cardinality and self-duality. 
\begin{corollary}\label{cor:card}
Let $\sC$ and $\sC^{\perp}$ be given as in Theorem \ref{thm:main1}. Then we obtain the equation
$$
\log_p|\sC|+\log_p|\sC^{\perp}|=\log_p|\sR|.
$$
\end{corollary}

\begin{proof}Follows from the equations (\ref{eq:sC}), (\ref{eq:sCperp}) and Lemma \ref{lem:cyclicness}. \end{proof}

\begin{corollary}\label{cor:selfadd}
Let $\sC$ be a code given as in (\ref{eq:sC}). Then, $\sC$ is self-dual if and only if $m$ is even and $e_{i,j}=\frac{m}{2}$ for all $0\leq i\leq v$ and $0\leq j\leq r-1$.
\end{corollary}

\begin{proof} The evenness of $m$ and the property $e_{i,j}=\frac{m}{2}$ for all $0\leq i\leq v$ and $0\leq j\leq r-1$ can be derived from (\ref{eq:sC}) and (\ref{eq:sCperp}). Remark that the basis $\{1,\omega,...,\omega^{r-1}\}$ is also the basis of $\Z_{p^n}[\omega]$ over $\Z_{p^n}$, hence the dual basis $\{\theta_0,...,\theta_{r-1}\}$ of it exists, and this existence is enough to complete the proof (since all $e_{i,j}$'s are the same).     
\end{proof}

Now, we illustrate Theorem \ref{thm:main1} and its corollaries above in the following example.
\begin{example}\label{ex:Gal}
Let $S=\Z_4[x]/<x^2+2,2x>$ and $R=S[\omega]$, where $\omega$ is a root of the polynomial $X^2+X+1\in S[X]$. That is, $p=2,n=2,r=2,g(x)=x^2+2,k=2,t=1$ and so $m=3$. Let also $N=3$ (which is relatively prime to $p$). Hence $\widehat{R}=R$ and $\zeta =\omega =\eta$. Then we have
$$
C_0^{(p)}=C_0^{(p^r)}=\{0\}, \quad C_1^{(p)}=\{1,2\},\quad C_{1,0}^{(p^r)}=\{1\}\text{ and }C_{1,1}^{(p^r)}=\{2\}. 
$$ 
That is, $u=0$ and $v=1$. Accordingly,
$$
\begin{array}{ll}
m_0(X)=X+3, & m_{1,0}(X)=X+3\omega, \\
m_1(X)=X^2+X+1, & m_{1,1}(X)=X+(\omega +1),	
\end{array}
$$
and 
$$
\begin{array}{ll}
\epsilon_0(X)=3X^2+3X+3, & \epsilon_{1,0}(X)=3\omega X^2+(\omega +1)X+3, \\
\epsilon_1(X)=X^2+X+2, & \epsilon_{1,1}(X)=(\omega +1)X^2+3\omega X+3.	
\end{array}
$$
Direct computations give that
$$
\begin{array}{l}
\sK_0 =\{aX^2+aX+a:a\in S\},  \\  
\sK_1 =\{aX^2+bX+(-a-b):a,b\in S\},  \\
\sK_{1,0} =\{(a+\omega b)X^2+(-b+\omega (a-b))X+(b-a+\omega(-a)):a,b\in S\}, \\	
\sK_{1,1} =\{(a+\omega b)X^2+(b-a+\omega (-a))X+(-b+\omega (a-b)):a,b\in S\}.	
\end{array}
$$
Notice also that $\theta_0=\omega +3$ and $\theta_1=2\omega +1$. Let 
$$
\sC =\sK_0 +2\omega \sK_0 + x\sK_{1,0},
$$
then
$$
\sC^{\perp} =x\theta_1\sK_0 + 2\sK_{1,0} + \sK_{1,1}.
$$
Observe that $|\sC|=2^8$, $|\sC^{\perp}|=2^{10}$ and $|\sR|=(2^6)^3$. Also remark that any self-dual codes do not exist in this ambient space $\sR$, since $m$ is not even.
\end{example}

\section{Eisenstein-additive cyclic codes}\label{sec:addEis}

Consider the same definitions and notations in Section \ref{sec:addGal} but inserting $R=\frac{\Z_{p^n}[x]}{<g(x),p^{n-1}x^t>}$ and $S=\Z_{p^n}$, where $p$ is a prime, $n$ is a positive integer and $g(x)\in\Z_{p^n}[x]$ is an Eisenstein polynomial of degree $k$. Notice that $S$ is a Galois ring which is the coefficient ring of $R$. In this case, we obtain $u=v$. 

Now, before to give our main theorem, we give a lemma to clarify some further points.

\begin{lemma}\label{lem:free}
$R$ is not a free module over $S$.
\end{lemma}
\begin{proof}
Clearly, the set $B=\{1,x,...,x^{k-1}\}$ is a minimal set spanning $R$ over $S$. However, $B$ is not linearly independent, since $p^{n-1}$ is non-zero but $p^{n-1}x^t=0$. Therefore, no bases of $R$ exist over $S$, i.e., $R$ is not a free module over $S$.  
\end{proof}

The second main result in this paper is the following theorem. 
\begin{theorem}\label{thm:main2}
Any $S$-linear cyclic code $\sC\subseteq\sR$ over $R$ of length $N$ is of the form
\begin{equation}
\sC =\sum_{i=0}^{v}\sum_{j=0}^{m-1}a_{i,j}x^j\sK_{i}, 
\label{eq:sC-Eis}
\end{equation}
for some $a_{i,j}\in\{0,1\}\subseteq R$. 
\end{theorem}

\begin{proof}
The proof can be done similar to the proof of Theorem \ref{thm:main1}. Remark that Lemma \ref{lem:mu} works efficiently also here, since $u=v$. 
\end{proof}

\begin{remark}\label{rem:aij}
Notice that, if $a_{i,j_0}=1$ for some $0\leq j_0\leq m-1$, then $a_{i,j_0+kl}$ can be taken both zero and one for $l\geq 1$ (since $<x^k>=<p>\subseteq S$). However, any such situation does not disturb the set up of Theorem \ref{thm:main2}. 
\end{remark}

\begin{remark}\label{rem:dualdual}
We have not mention duality in Theorem \ref{thm:main2}, whereas we have done in Theorem \ref{thm:main1}. The reason is related to the profile of the extension of $R$ over $S$. Since $R$ is not free over $S$, we can not determine any trace function for $x$ in $R$ over $S$, and hence we can not define any trace inner products. In the following subsection, we examine the duality notion for Eisenstein additivity separately.  
\end{remark}

\subsection{Character Theoretic Duality for Eisenstein-Additive Codes}\label{sec:EisDual}

Eisenstein extension is not a free extension when $t\neq k$, hence the problem of a suitable inner product for Eisenstein-additive codes occurs. The character theoretic approach in \cite{Woo2008} provides a convenient inner product and a duality notion when we consider the one-to-oneness between a code and its dual (a MacWilliams identity). Remark that the character theoretic duality notion in \cite{Woo2008} was given for Frobenius rings, so we may apply this notion to finite commutative chain rings. In this subsection, we adjust the notion in \cite{Woo2008} to our context assuming that the reader has some basic knowledge about characters (otherwise, we suggest \cite[Section 3]{Woo2008} for the sufficient information about characters we use in this paper).  
 
Consider commutative chain ring $R=\Z_{p^n}[x]/<g(x),p^{n-1}x^t>$, where $g(x)$ is an Eisenstein polynomial of degree $k$. Clearly the additive structure of $R$ is isomorphic to the finite abelian group $G$ of the form  
$$
G=\bigoplus_{i=1}^{t}\Z_{p^n}\oplus\bigoplus_{i=1}^{k-t}\Z_{p^{n-1}}.
$$
Consider the unique representation $a=a_0+a_1x+...+a_{k-1}x^{k-1}$ of elements $a\in R$, where $a_i\in\Z_{p^n}$ for $0\leq i\leq t-1$ and $a_i\in\{0,1,...,p^{n-1}-1\}\subseteq\Z_{p^{n}}$ for $t\leq i\leq k-1$. Corresponding to each element $a\in R$, we define a map $\chi_a$ from $R$ to $\C$ (the set of complex numbers) given by
$$
\chi_a:z=z_0+z_1x+...+z_{k-1}x^{k-1}\mapsto \eta_{p^n}^{a_0z_0+...+a_{t-1}z_{t-1}}\eta_{p^{n-1}}^{a_tz_t+...+a_{k-1}z_{k-1}},
$$
where $\eta_{p^n}$ and $\eta_{p^{n-1}}$ are the $p^n$.th and the $p^{n-1}$.th root of unities respectively. This map is clearly an additive character of $R$ (that is, $\chi_a$ is a group homomorphism from the additive structure of $R$ to the group of non-zero complex numbers with multiplication). Let $\chi=\{\chi_a: a\in R\}$, then $\chi$ with the point-wise multiplication is a group and isomorphic to the additive structure of $R$ (i.e. $\chi$ is a (additive) character group of $R$).         

In addition, $\chi_a$ corresponds to the vector $(\chi_a(z))_{z\in R}\in\C^{|R|}$ for each $a\in R$. Remark that the addition $a+b$ in $R$ corresponds to the component-wise multiplication of $(\chi_a(z))_{z\in R}$ and $(\chi_b(z))_{z\in R}$, when we fix the order of elements $z$ in $R$. We use the notation $\chi_a$ for both the homomorphism $\chi_a$ and the corresponding vector $(\chi_a(z))_{z\in R}$ when the meaning is clear. Now we define an inner product between vectors $\chi_a$ and $\chi_b$ by
\begin{equation}\label{eq:Herm}
(\chi_a,\chi_b)=\frac{1}{|R|}\sum_{z\in R}\chi_a(z)\overline{\chi_b(z)},
\end{equation}
where $\overline{\chi_b(z)}$ denotes the complex conjugate of $\chi_b(z)$. This inner product is indeed a positive definite Hermitian product.

We consider ring $S=\Z_{p^n}$ as an additive subgroup of $R$, and hence define the \textit{annihilator} $(\chi :S)=\{\chi_a\in\chi: \chi_a(z)=0\text{ for all }z\in S\}$. Then $(\chi :S)$ is isomorphic to the character group of the quotient group $R/S$, and hence $|(\chi :S)|=|R|/|S|$. 

Let $G_1$ and $G_2$ be two finite abelian groups and $\chi_1$ and $\chi_2$ be their character groups respectively. Then the character group of the group $G_1\times G_2$ with the component-wise operations is indeed $\chi_1\times \chi_2$. 

All the set up about characters mentioned above allow us to define dual codes for Eisenstein-additive codes. An Eisenstein-additive code $\sC$ was an additive subgroup of $R^N$ (or equivalently $\sR$), thus the \textit{dual} of $\sC$ is defined as the additive subgroup $\sD$ of $R^N$ given by $\sD=\{(a_1,...,a_N)\in R^N: (\chi_{a_1},...,\chi_{a_N})\in (\chi^N:\sC)\}$. In that way, the one-to-oneness between an Eisenstein-additive code and its dual is satisfied according to \cite[Theorem 4.2.1]{Woo2008}.

\begin{example}\label{ex:Eis}
	Let $S=\Z_4$ and $R=S[x]/<x^2+2,2x>$. That is, $p=2,n=2,r=1,g(x)=x^2+2,k=2,t=1$ and so $m=3$. Let also $N=3$ (which is relatively prime to $p$). Hence $\widehat{R}=R$ and $\zeta =\eta$. Then we have
	$$
	\begin{array}{ll}
	C_0^{(p)}=\{0\}, 	& 	C_1^{(p)}=\{1,2\}; 				\\
	m_0(X)=X+3, 			& 	m_1(X)=X^2+X+1;						\\
	\epsilon_0(X)=3X^2+3X+3, & \epsilon_1(X)=X^2+X+2.
	\end{array}
	$$
	and hence
	$$
	\begin{array}{ll}
	\sK_0 =\{aX^2+aX+a:a\in S\}, & \sK_1 =\{aX^2+bX+(-a-b):a,b\in S\}.
	\end{array}
	$$
If we write the elements $a\in R$ as $a=a_0+a_1x$ where $a_0\in\Z_4$ and $a_1\in\{0,1\}$, then we may define the corresponding characters $\chi_a$ as $\chi_a=i^{a_0}(-1)^{a_1}$ where $i$ is the primitive $4$.th root of unity (in the set of complex numbers). Considering the ordering 
$$
(0,1,2,3,x,x+1,x+2,x+3)
$$
of elements in $R$, we may write the elements of the character group $\chi$ of $R$ in vector form as follows:
$$
\begin{array}{lll}
\chi_0 & = & (1,1,1,1,1,1,1,1), \\
\chi_1 & = & (1,i,-1,-i,1,i,-1,-i), \\
\chi_2 & = & (1,-1,1,-1,1,-1,1,-1), \\
\chi_3 & = & (1,-i,-1,i,1,-i,-1,i), \\
\chi_x & = & (1,1,1,1,-1,-1,-1,-1), \\
\chi_{x+1} & = & (1,i,-1,-i,-1,-i,1,i), \\
\chi_{x+2} & = & (1,-1,1,-1,-1,1,-1,1), \\
\chi_{x+3} & = & (1,-i,-1,i,-1,i,1,-i).
\end{array}
$$
Then the annihilator of $S$ is $(\chi :S)=\{\chi_0,\chi_x\}$ and hence the dual of $S$ is $xS=\{0,x\}$. Duality of other subgroups of $R$ is as follows.
 $$
 \begin{array}{rll}
	R & \leftrightarrow & \{0\}, \\
	S & \leftrightarrow & xS, \\
	xR & \leftrightarrow & 2R. 
 \end{array}
 $$
Now, let us define a code. Let 
$$
\sC =\sK_0+x\sK_0 + x\sK_1.
$$
Here, we can also write $\sC =\sK_0+x\sK_0+x^2\sK_0+ x\sK_1$ (recall Remark \ref{rem:aij}). In addition, we can write $\sC=\sL_0+x\sK_1=R\sK_0+x\sK_1$. Then we have
	$$
	\sC^{\perp} =\sK_1.
	$$
Observe that $|\sC|=2^5$ and $|\sC^{\perp}|=2^{4}$, i.e., $|\sC|\cdot|\sC^{\perp}|=|\sR|=2^9$. Also remark that any self-dual codes do not exist in this ambient space, since no subgroups of $R$ are self-dual.
\end{example}


\section*{Acknowledgement}
The skeleton of this study has been constructed during the second author's visit to University of Valladolid between March 1-31 in 2016 by the support of COST Action IC 1104 Random Network Coding and Designs over GF(q). 

\end{document}